\documentclass[cmp]{svjour}
\usepackage{latexsym}
\usepackage{amssymb}
\usepackage{amsmath}

\usepackage[dvips]{graphicx}

\allowdisplaybreaks

\newtheorem{thm}    {Theorem}
\newtheorem{lem}     {Lemma}
\newtheorem{condition}  {Condition}

\def\mix{\mathop{\rm mix}}

\newcommand{\bR}{\mathbb{R}}

\newcommand{\bF}{\mathbb{F}}

\def\cA{{\cal A}}

\def\rE{{\rm E}}

\newcommand{\Tr}{{\rm Tr}\,}

\newcommand{\bX}{{\bf X}}

\makeatletter
\newcommand{\lleq}{\mathrel{\mathpalette\gl@align<}}
\newcommand{\ggeq}{\mathrel{\mathpalette\gl@align>}}
\newcommand{\gl@align}[2]{
\vbox{\baselineskip\z@skip\lineskip\z@
\ialign{$\m@th#1\hfil##\hfil$\crcr#2\crcr{}_{{}_{(=)}}\crcr}}}
\makeatother

\def\Label#1{\label{#1}\ [\ \text{#1}\ ]\ }
\def\Label{\label}

\begin{document}

\title{{Precise evaluation of leaked information with universal$_2$ privacy amplification in the presence of quantum attacker}}
\titlerunning{Precise evaluation of leaked information}

\author{Masahito Hayashi$^{1,2}$}
\institute{$^{1}$~Graduate School of Mathematics, Nagoya University, Japan. \\
$^{2}$~Centre for Quantum Technologies, National University of Singapore, Singapore. \\
\email{masahito@math.nagoya-u.ac.jp}}
\authorrunning{Masahito Hayashi}

\date{Received:}

\maketitle

\begin{abstract}
We treat 
secret key extraction when the eavesdropper has correlated quantum states.
We propose quantum privacy amplification theorems different from Renner's,
which are based on quantum conditional R\'{e}nyi entropy of order $1+s$.
Using those theorems, we derive an exponential decreasing rate
for leaked information and the asymptotic equivocation rate,
which have not been derived hitherto in the quantum setting.
\end{abstract}


\section{Introduction}
Extracting secret random numbers in the presence of quantum attacker
is one of important topics in quantum information theory.
The classical version of this topic was discussed by \cite{AC93,BBCM,ILL,RW,H-leaked}.
The quantum version is mainly treated by Renner\cite{Ren05} and his collaborators with using universal$_2$ hash function.
Indeed, a universal$_2$ hash function can be implemented with small amount of calculation.
As is shown by Renner\cite{Ren05}, when the classical random variable is correlated with eavesdropper's quantum state,
applying a universal$_2$ hash function as a privacy amplification process,
we obtain a secret random variable.


When the size of generated final random variable is sufficiently small, 
the final bits are almost independent of eavesdropper's quantum state.
Then, it is needed to evaluate the leaked information 
of the protocol using a universal$_2$ hash function.
In order to evaluate the secrecy, 
Renner\cite{Ren05} showed a privacy amplification theorem under the trace norm distance with the conditional R\'{e}nyi entropy of order 2.
Combining this theorem to the smoothing method, he provided the evaluation for the secrecy of final random variable.
However, application of smoothing has several difficulty.
In this paper, we derive another type of privacy amplification theorem
by using the conditional R\'{e}nyi entropy of order $1+s$.
Then, we can directly show the security without smoothing
when the final key size is smaller than the conditional entropy. 
That is, our proof is more direct as is mention below.

In this paper, 
we use a security criterion 
for leaked information and the difference from the uniform distribution by modifying the quantum mutual information.
Using the conditional R\'{e}nyi entropy of order $1+s$,
we propose other types of privacy amplification theorems under the above criterion.
The fundamental theorem for this purpose is derived by extending 
classical privacy amplification theorems obtained by \cite{H-leaked,network}.
Using one of these theorems, we derive an exponential decreasing rate of the criterion.
That is, when the extracted key rate is less than the conditional entropy, the criterion goes to zero exponentially.
Then, we derive an exponential decreasing rate for leaked information,
whose commutative case is the same as that by
\cite{H-leaked}.
Our derivation is contrastive with \cite{H-cq} 
in the point that our method does not employ smoothing method.
Our exponent is better than that given in \cite{H-cq}
under the modified quantum mutual information criterion.
Further, using the Pinsker inequality, we apply our result to the criterion for the trace norm distance.

When the extracted key rate is larger than the conditional entropy, 
the leaked information does not go to zero.
In this case, we focus on the minimum conditional entropy rate.
The rate is called the equivocation rate \cite{Wyner}
and the quantum version has not been treated until now.
Then, we derive the equivocation rate as by treating the minimum leaked information rate.
The smoothing method cannot evaluate the leaked information rate in this case
while the smoothing method can derive lower bounds for exponential decreasing rate\cite{H-cq}.
Since our method directly evaluate the information amount leaked to the eavesdropper,
it enable us to derive the equivocation rate.

This paper is organized as follows.
In Section \ref{s2}, we prepare quantum versions of information quantities.
In Section \ref{s3}, we 
formulate our setting and 
derive the exponents of leaked information when
the key generation rate is less than the conditional entropy rate.
In Section \ref{s7},
we compare our exponents with the exponents given by the smoothing method in 
\cite{H-cq}.
In Section \ref{s4}, we derive
the equivocation rate as the minimum conditional entropy
for a given key generation rate.
The proofs for Theorem \ref{th1} and \ref{th2} are given 
in Appendix.

\section{Information quantities}\Label{s2}
In order to treat 
leaked information 
after universal$_2$ privacy amplification
in the quantum setting,
we prepare several information quantities in a composite system ${\cal H}_a \otimes {\cal H}_E$,
in which,  ${\cal H}_a$ is a classical system spanned by the basis $\{|a\rangle\}$.
When the composite state is $
\rho=
\sum_a P(a)|a\rangle \langle a| \otimes \rho_a$,
the von Neumann entropies 
and 
Renyi entropies
are given as
\begin{align*}
H(A,E|\rho) &:= -\Tr \rho \log \rho \\
H(E|\rho) &:= -\Tr \rho^E \log \rho^E \\
H_{1+s}(A,E|\rho) &:=\frac{-1}{s}\log \Tr \rho^{1+s}  \\
H_{1+s}(E|\rho) &:=\frac{-1}{s}\log \Tr (\rho^E)^{1+s} 
\end{align*}
with $s\in \bR$
and $\rho^E= \Tr_A\rho$.
When we focus on the total system of a given density $\rho$,
$H(A,E|\rho) $ and $H_{1+s}(A,E|\rho)$
are simplified to 
$H(\rho)$ and $H_{1+s}(\rho)$.

Two kinds of quantum versions of 
the conditional entropy and conditional Renyi entropy 
are given for $s \in \bR$:
\begin{align*}
H(A|E|\rho) &:= H(A,E|\rho)-H(E|\rho) \\
\overline{H}(A|E|\rho) &:= -\Tr \rho \log (I_A \otimes  (\rho^E)^{-1/2} \rho I_A \otimes (\rho^E)^{-1/2}) \\
H_{1+s}(A|E|\rho) &:=\frac{-1}{s} \log \Tr \rho^{1+s} I_A \otimes  (\rho^E)^{-s} \\
\overline{H}_{1+s}^*(A|E|\rho) &:=\frac{-1}{s} 
\log \Tr \rho (I_A \otimes (\rho^E)^{-1/2} \rho I_A \otimes (\rho^E)^{-1/2})^s .
\end{align*}
The quantity $H_{1+s}(A|E|\rho)$ is used for the 
exponential decreasing rate for the security criterion
in Section III
while
$\overline{H}_{1+s}^*(A|E|\rho)$ is used for 
our derivation of the equivocation rate
in Section IV.
Indeed, while 
the quantity $\overline{H}_{2}^*(A|E|\rho) $
is the same as 
the quantity $H_2(A|E|\rho)$ given in \cite{Ren05}
and
the quantity $\overline{H}_2(A|E|\rho)$ given in \cite{H-cq},
the quantity $\overline{H}_{1+s}^*(A|E|\rho) $
is different from 
the quantity $\overline{H}_{1+s}(A|E|\rho) $ given in \cite{H-cq}
with $0<s<1$.

Since the functions $s \mapsto s H_{1+s}(A|E|\rho)$ 
and
\par\noindent
$s \mapsto s \overline{H}^*_{1+s}(A|E|\rho)$
are concave
and
$0{H}_{1}(A|E|\rho)=0 \overline{H}_{1}(A|E|\rho)
=0$,
$ {H}_{1+s}(A|E|\rho)$ and $\overline{H}^*_{1+s}(A|E|\rho)$ are monotone decreasing for $s\in \bR$.
Since $\lim_{s \to \infty} \overline{H}^*_{1+s}(A|E|\rho)$ coincides with the min entropy 
$H_{\min}(A|E|\rho) :=- \log \| 
I_A \otimes (\rho^E)^{-1/2} \rho I_A \otimes (\rho^E)^{-1/2}
\|$,
$
\overline{H}^*_{1+s}(A|E|\rho)
\ge 
H_{\min}(A|E|\rho)$.
Since 
$ {H}_{2}(A|E|\rho)=
- \log \Tr \rho 
(\rho^{1/2} (I_A \otimes \rho^E)^{-1}\rho^{1/2})$,
we have
$ 
{H}_{1+s}(A|E|\rho)
\ge {H}_{2}(A|E|\rho)
\ge -
\log \|(\rho^{1/2} (I_A \otimes \rho^E)^{-1}\rho^{1/2})\|
=
-\log \|
I_A \otimes (\rho^E)^{-1/2} \rho I_A \otimes (\rho^E)^{-1/2}
\|=H_{\min}(A|E|\rho) $
for $s\in (0,1]$.
Further, since $\lim_{s \to 0}H_{1+s}(A|E|\rho)=H(A|E|\rho)$
and
$\lim_{s \to 0}\overline{H}^*_{1+s}(A|E|\rho)=\overline{H}(A|E|\rho)$,
we have 
\begin{align}
H(A|E|\rho) & \ge H_{1+s}(A|E|\rho), \\
\overline{H}(A|E|\rho) & \ge \overline{H}^*_{1+s}(A|E|\rho)
\Label{8-15-14} 
\end{align}
for $s\in (0,1]$.


Then, the correlation between $A$ and ${\cal H}_E$
can be evaluated by 
two kinds of quantum versions of
the mutual information
\begin{align}
I(A:E|\rho) &:= D( \rho \| \rho_A \otimes \rho^E) \\
\underline{I}(A:E|\rho) &:= \underline{D}( \rho \| \rho_A \otimes \rho^E) \\
D(\rho\|\sigma) &:= \Tr \rho (\log \rho-\log \sigma) \\
\underline{D}(\rho\|\sigma) &:= \Tr \rho \log (\sigma^{-1/2} \rho \sigma^{-1/2}) .
\end{align}
By using the completely mixed state $\rho_{\mix}^A$ on ${\cal A}$,
two kinds of quantum versions of the mutual information can be modified to
\begin{align}
I'(A:E|\rho) &:= D( \rho \| \rho_{\mix}^A \otimes \rho^E) \nonumber \\
&=I(A:E|\rho)+ D(\rho^A\|\rho_{\mix}^A )\\
&=I(A:E|\rho)+ H(A|\rho_{\mix}^A)-H(A|\rho^A)\\
\underline{I}'(A:E|\rho) &:= \underline{D}( \rho \| \rho_{\mix}^A \otimes \rho^E) ,
\end{align}
which satisfy
\begin{align*}
I(A:E|\rho) & \le I'(A:E|\rho) \\
\underline{I}(A:E|\rho) & \le \underline{I}'(A:E|\rho) 
\end{align*}
and
\begin{align}
H(A|E|\rho) & = -I'(A:E|\rho) +\log |{\cal A}| \Label{1-28-1}\\
\overline{H}(A|E|\rho) & = -\underline{I}'(A:E|\rho) +\log |{\cal A}| \Label{8-15-17} .
\end{align}

Indeed, the quantity $I(A:E|\rho^{A,E})$ 
represents the amount of information leaked to $E$,
and the remaining quantity $D(\rho^A\|\rho_{\mix}^A )$
describes the difference of the random number $A$ from the uniform random number.
So, if 
the quantity $I'(A:E|\rho^{A,E})$ is small,
we can conclude that the random number $A$ has less correlation with $E$ 
and is close to the uniform random number.
In particular,
if the quantity $I'(A:E|\rho^{A,E})$ goes to zero,
the mutual information $I(A:E | \rho^{A,E})$ 
goes to zero,
and the state $\rho^{A}$ goes to 
the completely mixed state $\rho_{\mix}^A$.
Hence, we can adopt 
the quantity $I'(A:E|\rho^{A,E})$ as a criterion for qualifying the secret random number.

Using the trace norm, we can evaluate the secrecy for the state $\rho$ as follows:
\begin{align}
d_1(A:E|\rho):=\| \rho -\rho^A \otimes \rho^E \|_1.
\end{align}
Taking into account the randomness, 
Renner \cite{Ren05} defined the following criteria for security of a secret random number:
\begin{align}
d_1'(A:E|\rho):=
\| \rho -\rho_{\mix}^A \otimes \rho^E \|_1.
\end{align}
Using the quantum version of Pinsker inequality,
we obtain
\begin{align}
d_1(A:E|\rho )^2  &\le I(A:E|\rho) \Label{8-19-14-a} \\
d_1'(A:E|\rho )^2 &\le I'(A:E|\rho).\Label{8-19-14}
\end{align}

When we apply the function $f$ to the classical random number $a \in \cA$,
$H(f(A),E|\rho) \le H(A,E|\rho)$, i.e., 
\begin{align}
H(f(A)|E|\rho) \le H(A|E|\rho).\Label{8-14-1}
\end{align}

As is shown in \cite{H-cq},
when we apply 
a quantum operation ${\cal E}$ on ${\cal H}_E$, 
since
it does not act on the classical system ${\cal A}$,
\begin{align}
H(A|E|{\cal E}(\rho)) &\ge H(A|E|\rho) \\
H_{1+s}(A|E|{\cal E}(\rho)) &\ge H_{1+s}(A|E|\rho). \Label{8-15-12} 
\end{align}

When the state $\sigma$
has the spectral decomposition $\sigma= \sum_i s_i E_i$,
the pinching map ${\cal E}_\sigma$ is defined as
\begin{align}
{\cal E}_\sigma(\rho):=\sum_{i} E_i \rho E_i.
\end{align}
When $v$ is the number of the eigenvalues of $\sigma$,
the inequality 
\begin{align}
\rho \le v {\cal E}_{\sigma}(\rho)
\Label{8-15-23}
\end{align}
holds\cite[Lemma 3.8]{Hayashi-book},\cite{H2001}.
Hence, we obtain
\begin{align}
\sigma^{-1/2} \rho \sigma^{-1/2}
\le
v \sigma^{-1/2} {\cal E}_{\sigma}(\rho)  \sigma^{-1/2}. \Label{8-26-3}
\end{align}
As $x \mapsto \log x$ is matrix monotone,
\begin{align}
\log \sigma^{-1/2} \rho \sigma^{-1/2}
\le
\log v+ 
\log \sigma^{-1/2} {\cal E}_{\sigma}(\rho)  \sigma^{-1/2} .
\end{align}
Since
\begin{align}
\Tr \rho \log \sigma^{-1/2} {\cal E}_{\sigma}(\rho)  \sigma^{-1/2}
=\Tr {\cal E}_{\sigma}(\rho) \log \sigma^{-1/2} {\cal E}_{\sigma}(\rho)  \sigma^{-1/2} ,
\end{align}
we obtain 
\begin{align}
D(\rho\|\sigma)
\le 
D({\cal E}_{\sigma}(\rho)\|\sigma) +\log v=
\underline{D}({\cal E}_{\sigma}(\rho)\|\sigma) +\log v .
\Label{8-15-8-a}
\end{align}
Therefore, when $v$ is the number of the eigenvalues of 
$\rho^E:=\sum_a p(a)\rho_a^E$,
an inequality 
\begin{align}
I(A:E|\rho) 
& \le 
I(A:E|{\cal E}_{\rho^E}(\rho)) +\log v \nonumber \\
&= \underline{I}(A:E|{\cal E}_{\rho^E}(\rho)) +\log v 
\Label{8-15-8} 
\end{align}
holds.
Using these relations, we can show the following lemma.
\begin{lem}
\begin{align}
\overline{H}^*_{1+s}(A|E|\rho)
\ge H_{1+s}(A|E|\rho).\Label{8-29-20}
\end{align}
\end{lem}

\begin{proof}
Applying (\ref{8-26-3}) to the case of $\sigma= \rho^E$,
we obtain
\begin{align*}
(\rho^E)^{-1/2} 
\rho
(\rho^E)^{-1/2} 
\le 
v 
(\rho^E)^{-1/2} {\cal E}_{\rho^E}(\rho) (\rho^E)^{-1/2} .
\end{align*}
Since $x \to x^s$ is matrix monotone,
we obtain
\begin{align*}
[(\rho^E)^{-1/2} 
\rho
(\rho^E)^{-1/2} ]^s
\le 
v^s 
[(\rho^E)^{-1/2} {\cal E}_{\rho^E}(\rho) (\rho^E)^{-1/2} ]^s.
\end{align*}
Hence,
\begin{align*}
& e^{-s\overline{H}^*_{1+s}(A|E|\rho)}
= \Tr \rho
[(\rho^E)^{-1/2} 
\rho
(\rho^E)^{-1/2} ]^s \\
\le &
v^s 
\Tr \rho
[(\rho^E)^{-1/2} {\cal E}_{\rho^E}(\rho) (\rho^E)^{-1/2} ]^s\\
=&
v^s 
\Tr 
{\cal E}_{\rho^E}(\rho)
[(\rho^E)^{-1/2} {\cal E}_{\rho^E}(\rho) (\rho^E)^{-1/2} ]^s
=
v^s  e^{-s\overline{H}^*_{1+s}(A|E|{\cal E}_{\rho^E}(\rho) )}\\
= & 
v^s  e^{-s H_{1+s}(A|E|{\cal E}_{\rho^E}(\rho) )}
\le 
v^s  e^{-s H_{1+s}(A|E|\rho )},
\end{align*}
where (\ref{8-15-12}) is used in the final inequality.
Letting $v_n$ be the number of eigenvalues of $(\rho^E)^{\otimes n}$,
we obtain
\begin{align*}
& n \overline{H}^*_{1+s}(A|E|\rho)
+\frac{\log v_n^s}{s}
=
\overline{H}^*_{1+s}(A|E|\rho^{\otimes n})
+\frac{\log v_n^s}{s} \\
\ge  &
H_{1+s}(A|E|\rho^{\otimes n} )
=n H_{1+s}(A|E|\rho ).
\end{align*}
Taking the limit $n \to \infty$,
we obtain (\ref{8-29-20}).
\end{proof}

\section{Formulation and exponential decreasing rate}\Label{s3}
We consider the secure key extraction problem from
a common classical random number $a \in \cA$ which has been partially eavesdropped as quantum states by Eve.
For this problem, it is assumed that Alice and Bob share a common classical random number $a \in \cA$,
and Eve has a quantum state $\rho_a \in {\cal H}_E$, which is correlated to the random number $a$. 
The task is to extract a common random number 
$f(a)$ from the random number $a \in \cA$, which is almost independent of 
Eve's quantum state.
Here, Alice and Bob are only allowed to apply the same function $f$ to the common random number $a \in \cA$ as Fig. \ref{f1}.
Now, we focus on an ensemble of the functions $f_{\bX}$ from 
$\cA$ to $\{1, \ldots, M\}$, where $\bX$ denotes a random variable describing 
the stochastic behavior of the function $f$.
An ensemble of the functions $f_{\bX}$ is called universal$_2$ 
when it satisfies the following condition\cite{Carter}:
\begin{condition}\Label{C1}
$\forall a_1 \neq \forall a_2\in \cA$,
the probability that $f_{\bX}(a_1)=f_{\bX}(a_2)$ is
at most $\frac{1}{M}$.
\end{condition}

\begin{figure}[htbp]
\begin{center}
\scalebox{0.4}{\includegraphics[scale=0.8]{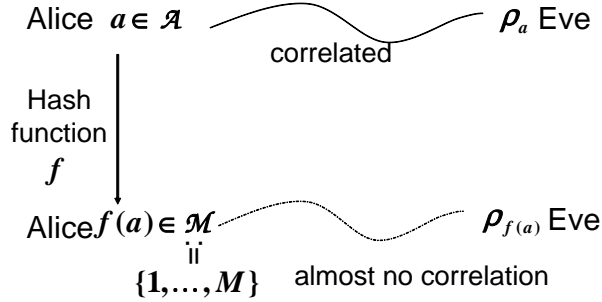}}
\end{center}
\caption{
Application of hash function}
\Label{f1}
\end{figure}%

Indeed, when the cardinality $|\cA|$ is a power of a prime power $q$
and $M$ is another power of the same prime power $q$,
an ensemble $\{f_{\bX}\}$ 
satisfying the both conditions
is given by the 
the concatenation of Toeplitz matrix and the identity 
$(\bX,I)$\cite{Krawczyk}
only with $\log_q |\cA|-1$ random variables taking values in the finite filed $\bF_q$.
That is, the matrix $(\bX,I)$ has small complexity.

\begin{thm}\Label{th1}
When the ensemble of the functions $\{f_{\bX}\}$ is 
universal$_2$, it satisfies 
\begin{align}
& 
I(f_{\bX}(A):E,\bX|\rho,P^{\bX}) 
\le 
I'(f_{\bX}(A):E,\bX|\rho,P^{\bX}) 
=\rE_\bX 
I'(f_{\bX}(A):E|\rho) \nonumber \\
\le &
\frac{v^s M^s}{s} e^{- s{H}_{1+s}(A|E|\rho)} 
=  v^s \frac{e^{s(\log M-H_{1+s}(A|E|\rho))}}{s},
\Label{8-14-2}
\end{align}
where $v$ is the number of eigenvalues of $\rho^E$.
\end{thm}

That is, there exists a function $f:{\cal A}\to \{1,\ldots, M\}$ such that
\begin{align}
I'(f(A):E|\rho) 
\le v^s \frac{e^{s(\log M-H_{1+s}(A|E|\rho))}}{s}.
\end{align}


Next, we consider the case when our state 
is given by the $n$-fold independent and identical state 
$\rho$, i.e., $\rho^{\otimes n}$.
We define the optimal generation rate
\begin{align*}
&G(\rho) \\
:=&
\sup_{\{(f_n,M_n)\}}
\left\{
\lim_{n\to\infty} \frac{\log M_n}{n}
\left|
\!
\begin{array}{l}
\displaystyle
\lim_{n\to\infty} \frac{I(f_n(A):E|\rho^{\otimes n})}{n}=0 \\
\displaystyle
\lim_{n\to\infty} \frac{H(f_n(A)|\rho^{\otimes n} )}{\log M_n}=1
\end{array}
\! 
\right. \right\} \\
=&
\sup_{\{(f_n,M_n)\}}
\left\{
\lim_{n\to\infty} \frac{\log M_n}{n}
\left|
\lim_{n\to\infty} \frac{I'(f_n(A):E|\rho^{\otimes n})}{n}=0
\right. \right\} ,
\end{align*}
whose classical version is treated by \cite{AC93}.
The second equation holds as follows.
the condition $
\lim_{n\to\infty} \frac{H(f_n(A)|\rho^{\otimes n} )}{\log M_n}=1$
is equivalent with 
$ \lim_{n\to\infty} \frac{D(\rho^{f_n(A)} \| \rho^{f_n(A)}_{\mix} ) }{n}=0$.
Hence, 
$\lim_{n\to\infty} \frac{I(f_n(A):E|\rho^{\otimes n})}{n}=0 $
and
$\lim_{n\to\infty} \frac{H(f_n(A)|\rho^{\otimes n} )}{\log M_n}=1$
if and only if
$\lim_{n\to\infty} \frac{I'(f_n(A):E|\rho^{\otimes n})}{n}=0$.

When the generation rate $R= \lim_{n\to\infty} \frac{\log M_n}{n}$ is smaller than $H(A|E)$,
there exists a sequence of functions $f_n:{\cal A}\to \{1,\ldots, e^{nR} \}$ such that
\begin{align}
I'(f_n(A):E|\rho^{\otimes n})
\le v_n^s \frac{e^{s(R-H_{1+s}(A|E|\rho^{\otimes n} ))}}{s},\Label{8-14-4}
\end{align}
where $v_n$ is the number of eigenvalues of $(\rho^E)^{\otimes n}$,
which is a polynomial increasing for $n$.
Since $\lim_{s \to 0}H_{1+s}(A|E|\rho))=H(A|E|\rho))$,
there exists a number $s \in (0,1]$ such that
$s(R-H_{1+s}(A|E|\rho))>0$.
Thus, the right hand side of (\ref{8-14-4}) goes to zero exponentially.
Conversely,
due to (\ref{8-14-1}),
any sequence of functions $f_n: {\cal A}^n \mapsto \{1, \ldots, e^{nR} \}$ satisfies that
\begin{align}
\lim_{n \to \infty} \frac{H(f_{n}(A)|E|\rho^{\otimes n})}{n} \le 
\frac{H(A|E|\rho^{\otimes n})}{n} = H(A|E|\rho).\Label{8-14-5}
\end{align}
When
$\lim_{n\to\infty} \frac{H(f_n(A)|\rho^{\otimes n} )}{nR}=1$,
\begin{align}
\lim_{n \to \infty} 
\frac{I(f_{n}(A):E|\rho^{\otimes n})}{n} 
&= R- \lim_{n \to \infty} \frac{H(f_{n}(A)|E|\rho^{\otimes n})}{n} \nonumber \\
& \ge 
R- H(A|E|\rho). 
\end{align}
That is, when $R>H(A|E|\rho)$,
$\frac{I(f_{n}(A):E|\rho^{\otimes n})}{n} $ does not go to zero.
Hence, we obtain
\begin{align}
G(\rho) =H(A|E|\rho)
\end{align}

In order to treat the speed of this convergence,
we focus on the supremum of  
the {\it exponentially decreasing rate (exponent)} of 
$I'(f_n(A):E|\rho^{\otimes n})$ for a given $R$
\begin{align*}
& e_I(\rho | R) \\
:=& \!\!\!
\sup_{\{(f_n,M_n)\}}\!\!
\left\{\!
\lim_{n\to\infty} \! \frac{-\log I'(f_n(A):E|\rho^{\otimes n})}{n}
\!\left|
\lim_{n\to\infty} \!\!\frac{-\log M_n}{n} 
\!\le\! R\!
\right. \right\}.
\end{align*}
Since the relation
$s{H}_{1+s}(A|E|\rho^{\otimes n})= n s{H}_{1+s}(A|E|\rho)$ holds,
the inequality (\ref{8-14-4}) implies that
\begin{align}
e_I(\rho|R) 
&\ge 
e_H(\rho|R):=
\max_{0 \le s \le 1}  s{H}_{1+s}(A|E|\rho)-sR \nonumber \\
&= \max_{0 \le s \le 1}  s (H_{1+s}(A|E|\rho )-R),
\Label{4-16-4}
\end{align}
whose commutative version coincides with the bound given in \cite{H-leaked}.

Next, we apply our evaluation to the criterion $d_1'(A:E|\rho)$.
When $\{f_{\bX}\}$ satisfies Condition \ref{C1},
combining (\ref{8-19-14}) and (\ref{8-14-2}),
we obtain
\begin{align}
\rE_{\bX} d_1'(f_{\bX}(A):E|\rho)
\le &
\sqrt{\rE_{\bX} d_1'(f_{\bX}(A):E|\rho)^2} \nonumber \\
\le &
\frac{v^{s/2} M^{s/2}}{\sqrt{s}} e^{- \frac{s}{2}H_{1+s}(A|E|\rho)} .
\Label{8-26-5}
\end{align}
That is, in the $n$-fold asymptotic setting, when the generation key rate is $R$,
we focus on the supremum of  
the {\it exponentially decreasing rate (exponent)} of 
$I(f_n(A):E|\rho^{\otimes n})$ for a given $R$
\begin{align*}
& e_d(\rho | R) \\
:=& \!\!\!
\sup_{\{(f_n,M_n)\}}\!\!
\left\{\!
\lim_{n\to\infty} \! \frac{-\log d_1'(f_n(A):E|\rho^{\otimes n})}{n}
\!\left|\!
\lim_{n\to\infty} \!\!\frac{-\log M_n}{n} 
\!\le\! R\!
\right. \right\}.
\end{align*}
Then, 
the inequality (\ref{8-26-5}) implies that
$e_d(\rho|R) \ge \frac{e_H(\rho|R)}{2}$,
whose commutative version is smaller than 
the bound given in \cite{H-tight}.

\section{Comparison with smoothing method}\Label{s7}
The paper \cite{H-cq} derived lower bounds
for $e_I(\rho|R)$ and $e_d(\rho|R)$.
In order describe them, 
we introduce an information quantity $\phi(s|A|E|\rho^{A,E})$:
\begin{align*}
\phi(s|A|E|\rho^{A,E})
&:=\log 
\Tr_E (\Tr_A 
(\rho^{A,E})^{1/(1-s)})^{1-s} \\
&=\log 
\Tr_E 
(\sum_a P^A(a)^{1/(1-s)} \rho_a^{1/(1-s)} )^{1-s} .
\end{align*}
This quantity satisfies the following lemma.
\begin{lem}\Label{l2-b}
\cite[Lemma 11]{H-cq}
The inequalities 
\begin{align}
s H_{1+s}(A|E|\rho^{A,E} ) 
\ge &
-\phi(s|A|E|\rho^{A,E})\Label{ineq-7-23-2} \\
s H_{1+s}(A|E|\rho^{A,E})
\le &
-(1+s)\phi(\frac{s}{1+s}|A|E|\rho^{A,E})
\Label{8-26-8}
\end{align}
hold for $0 \le s\le 1$.
\end{lem}
Then, the paper \cite{H-cq} showed that
\begin{align}
e_d(\rho|R) 
\ge & 
e_{\phi,q}(\rho^{A,E}|R) \\
e_I(\rho|R)
\ge &
e_{H,q}(\rho^{A,E}|R) \\
e_I(\rho|R)
\ge &
e_{\phi,q}(\rho^{A,E}|R),
\end{align}
where
\begin{align*}
e_{\phi,q}(\rho^{A,E}|R) 
:= &
\max_{0 \le s \le 1} 
-\frac{1+s}{2} \phi(\frac{s}{1+s}|\rho^{A,E})-\frac{s}{2}R \\
= & \max_{0 \le t \le \frac{1}{2}} 
-\frac{1}{2(1-t)}\phi(t|\rho^{A,E})-\frac{t}{2(1-t)}R \\
e_{H,q}(\rho^{A,E}|R) 
:=& \max_{0 \le s \le 1} \frac{s}{2-s} ( H_{1+s}(A|E|\rho^{A,E} ) -R). 
\end{align*}
As a relation, we obtain the following lemma.
\begin{lem}
\begin{align}
e_{H}(\rho|R)
\ge &
e_{H,q}(\rho^{A,E}|R) \Label{8-29-13}\\
e_{H}(\rho|R)
\ge &
e_{\phi,q}(\rho^{A,E}|R)\Label{8-29-14} \\
\frac{1}{2}e_{H}(\rho|R)
\le &
e_{\phi,q}(\rho^{A,E}|R) .\Label{8-29-15}
\end{align}
\end{lem}
Hence, our lower bound $e_{H}(\rho|R)$ for $e_I(\rho|R)$
is better than those given in \cite{H-cq}.
However,
our lower bound $\frac{1}{2}e_{H}(\rho|R)$ for $e_d(\rho|R)$
is not as good as that given in \cite{H-cq}.
This fact implies that
our method is better under the modified mutual information criterion
than the smoothing method used in \cite{H-cq}.

\begin{proof}
\begin{align*}
& e_H(\rho|R)
=\max_{0 \le s \le 1}  s (H_{1+s}(A|E|\rho )-R) \\
\ge & \max_{0 \le s \le 1}  \frac{s}{2-s} (H_{1+s}(A|E|\rho )-R) 
= e_{H,q}(\rho|R),
\end{align*}
which implies (\ref{8-29-13}).
Further, 
(\ref{ineq-7-23-2}) yields that
\begin{align*}
&e_{\phi,q}(\rho^{A,E}|R) 
= \max_{0 \le t \le \frac{1}{2}} 
-\frac{1}{2(1-t)}\phi(t|\rho^{A,E})-\frac{t}{2(1-t)}R \\
\le & \max_{0 \le t \le \frac{1}{2}} 
\frac{t}{2(1-t)}H_{1+t}(A|E|\rho ) -\frac{t}{2(1-t)}R \\
= & \max_{0 \le t \le \frac{1}{2}} 
\frac{t}{2(1-t)} (H_{1+t}(A|E|\rho ) -R) \\
= & \max_{0 \le t \le \frac{1}{2}} 
t (H_{1+t}(A|E|\rho ) -R) \\
\le & \max_{0 \le t \le 1} 
t (H_{1+t}(A|E|\rho ) -R) 
= e_H(\rho|R),
\end{align*}
which implies (\ref{8-29-14}).

Finally, 
(\ref{8-26-8})
yields that
\begin{align}
& \frac{1}{2}e_{H}(\rho^{A,E}|R)
=
\max_{0 \le s \le 1} \frac{s}{2} H_{1+s}(A|E|\rho^{A,E} ) -\frac{s}{2}  R 
\nonumber \\
\le &
\max_{0 \le s \le 1} - \frac{1+s}{2} \phi(\frac{s}{1+s}|A|E|\rho^{A,E} ) 
-\frac{s}{2}  R 
= e_{\phi,q}(\rho^{A,E}|R),\nonumber
\end{align}
which implies (\ref{8-29-15}).
\end{proof}

\section{Equivocation rate}\Label{s4}

Next, we consider the case when $\log M$ is larger than $H(A|E)$.

\begin{thm}\Label{th2}
When the ensemble of the functions $\{f_{\bX}\}$ is 
universal$_2$, it satisfies 
\begin{align}
\rE_\bX 
e^{s \underline{I}'(f_{\bX}(A):E|\rho) }
&\le
1+ M^s e^{- s \overline{H}_{1+s}^*(A|E|\rho)} \nonumber \\
&=
1+ e^{s (\log M- \overline{H}^*_{1+s}(A|E|\rho ))}.
\Label{8-14-3}
\end{align}
\end{thm}

Using (\ref{8-14-3}) and the concavity of $x \mapsto \log x$ , we obtain 
\begin{align*}
& s \rE_\bX 
\underline{I}'(f_{\bX}(A):E|\rho) 
\le 
\log \rE_\bX 
e^{s \underline{I}'(f_{\bX}(A):E|\rho) } \nonumber \\
\le &
\log (1+ e^{s (\log M- \overline{H}^*_{1+s}(A|E|\rho ))})
\le
e^{s (\log M- \overline{H}^*_{1+s}(A|E|\rho ))},
\end{align*}
which can be regarded as another version of (\ref{8-14-2}).

Hence, (\ref{8-15-8}), (\ref{8-14-3}), and (\ref{8-15-12}) guarantee that
\begin{align*}
& \rE_\bX 
e^{s I'(f_{\bX}(A):E|\rho) }
\le
v^s \rE_\bX 
e^{s I'(f_{\bX}(A):E|{\cal E}_{\rho^E}(\rho )) } \nonumber \\
\le &
v^s (1+ M^s e^{- s\overline{H}^*_{1+s}(A|E|{\cal E}_{\rho^E}(\rho ) )}) \nonumber \\
= &
v^s (1+ M^s e^{- s{H}_{1+s}(A|E|{\cal E}_{\rho^E}(\rho ) )}) \nonumber \\
\le &
v^s (1+ M^s e^{- s{H}_{1+s}(A|E|\rho)}) 
=
v^s (1+ e^{s (\log M- H_{1+s}(A|E|\rho ))} ),
\end{align*}
where $v$ is the number of eigenvalues of $\rho^E$.
Since
\begin{align*}
& \log v^s (1+ e^{s (\log M- H_{1+s}(A|E|\rho ))} ) \nonumber\\
=&
s \log v+\log  (1+ e^{s (\log M- H_{1+s}(A|E|\rho ))} ) \nonumber\\
\le &
s \log v+\log  2+ \log \max \{1, e^{s (\log M- H_{1+s}(A|E|\rho ))} \} \nonumber \\
=&
s \log v+\log  2+ \max \{0, s (\log M- H_{1+s}(A|E|\rho )) \},
\end{align*}
using (\ref{1-28-1}), we obtain the following theorem:
\begin{thm}
There exists a function $f: {\cal A} \mapsto \{1, \ldots, M\}$ such that
\begin{align*}
& \log M - H(f(A)|E|\rho )
=I'(f(A):E|\rho ) \nonumber \\
\le &
\log v+\frac{\log  2}{s}+ \max\{0, \log M- H_{1+s}(A|E|\rho ) \}.
\end{align*}
for $s\in (0,1]$.
\end{thm}

Next, we consider the case when our state
is given by the $n$-fold independent and identical state 
$\rho$, i.e., $\rho^{\otimes n}$.
Then, we define the {\it equivocation rate}
as the maximum Eve's ambiguity rate for the given key generation rate $R$:
\begin{align*}
 {\cal R}(R| \rho) 
:=
\sup_{\{f_n \}}
\{
\lim_{n \to \infty} \frac{H(f_{n}(A)|E|\rho^{\otimes n})}{n} 
|
\lim_{n\to\infty} \frac{H(f_n(A)|\rho^{\otimes n} )}{nR}=1
\},
\end{align*}
where the supremum takes the map $f_n$ that maps from ${\cal A}^n$ to $\{1,\ldots, e^{nR} \}$.
Then, we obtain the following theorem.
\begin{thm}
When the key generation rate $R$ is greater than 
$H(A|E|\rho)$,
\begin{align}
 {\cal R}(R| \rho) =H(A|E|\rho).\Label{1-28-2}
\end{align}
\end{thm}

Indeed, 
using the above theorem, we can calculate
the minimum information rate for the given key generation rate $R$
as follows.
\begin{align*}
&
\inf_{\{f_n \}}
\{
\lim_{n \to \infty} \frac{I(E:f_{n}(A)|\rho^{\otimes n})}{n} 
|
\lim_{n\to\infty} \frac{H(f_n(A)|\rho^{\otimes n} )}{nR}=1
\} \\
=& \max\{R- H(A|E|\rho),0\}.
\end{align*}

\begin{proof}
When the key generation rate $R$, i.e., $M_n=e^{nR}$,
there exists a sequence of functions $f_n: {\cal A}^n \mapsto \{1, \ldots, M_n\}$ such that
\begin{align*}
R- \lim_{n \to \infty} \frac{H(E|f_{n}(A)|\rho^{\otimes n})}{n} 
\le 
\max\{0, R- H_{1+s}(A|E|\rho ) \}
\end{align*}
for $s\in (0,1]$.
Then, taking the limit $s \to 0$,
we obtain
\begin{align*}
R- \lim_{n \to \infty} \frac{H(E|f_{n}(A)|\rho^{\otimes n})}{n} 
\le 
\max\{0, R- H(A|E|\rho ) \},
\end{align*}
which implies the part $\le$ of (\ref{1-28-2}).
Converse inequality $\ge$ follows from (\ref{8-14-5}).
\end{proof}

\section{Conclusion}
We have derived an upper bound of 
information leaked to quantum attacker 
in the modified quantum mutual information criterion
when we apply universal$_2$ hash functions.
In the commutative case, our lower bound coincides with 
the bound given in \cite{H-leaked}.
In the non-commutative case, 
our bound is different from Renner\cite{Ren05}'s two universal hashing lemma
even in $s=1$
because Renner\cite{Ren05}'s result is based on 
$\overline{H}_2^*(A|E|\rho)$ but ours is based on $H_{1+s}(A|E|\rho)$. 

Applying our bound to the i.i.d. case,
we obtain a lower bound for
the exponential decreasing rate for information leaked to quantum attacker
under the modified mutual information criterion.
Our lower bound is better than
lower bounds derived by the smoothing method in \cite{H-cq}.

Further, we have derived the asymptotic equivocation rate.
In oder to show it, 
we have derived a quantum version of 
privacy amplification theorems,
whose classical version is given in \cite{H-leaked,network}.
In this quantum version, we employ 
$\overline{H}_{1+s}^*(A|E|\rho)$ instead of $H_{1+s}(A|E|\rho)$. 
In the second step for the derivation,
we employ $H_{1+s}(A|E|\rho)$. 
Then, the asymptotic equivocation rate
can be characterized by $H(A|E|\rho)$,
which is given by the limit $\lim_{s \to 0} H_{1+s}(A|E|\rho)$. 

\section*{Acknowledgments}
The author 
is partially supported by a MEXT Grant-in-Aid for Young Scientists (A) No. 20686026 and Grant-in-Aid for Scientific Research (A) No. 23246071.
He is partially supported by the National Institute of Information and Communication Technolgy (NICT), Japan.
The Centre for Quantum Technologies is funded by the
Singapore Ministry of Education and the National Research Foundation
as part of the Research Centres of Excellence programme.

\appendix

\section{Proof of Theorem \ref{th1}}\Label{a1}
In order to show Theorem \ref{th1}, 
we prepare the following two lemmas.

\begin{lem}\Label{L9}
The matrix inequality $(I+X)^s \le I+X^s$ holds 
with a non-negative matrix $X$ and $s \in (0,1]$.
\end{lem}

\begin{proof}
Since $I$ is commutative with $X$,
it is sufficient to show that
$(1+x)^s \le 1+x^s$ for $x\ge 0$.
This inequality is trivial.
\end{proof}

\begin{lem}\Label{L10}
The matrix inequality $ \log (I+X) \le \frac{1}{s}X^s$ holds 
with a non-negative matrix $X$ and $s \in (0,1]$.
\end{lem}
\begin{proof}
Since $I$ is commutative with $X$,
it is sufficient to show that
$\log (1+x) \le \frac{x^s}{s}$ for $x\ge 0$.
Since the inequalities $(1+x)^s \le 1+ x^s$ 
and $\log (1+x) \le x$
hold for $x \ge 0$ and $0 < s \le 1$,
the inequalities
\begin{align}
\log (1+x) = \frac{\log (1+x)^s}{s}\le
\frac{\log (1+x^s)}{s}\le 
\frac{x^s}{s}\Label{4-27-14}
\end{align}
hold.
\end{proof}

Now, we prove Theorem \ref{th1}.
\begin{align}
& \rE_\bX I'(f_{\bX}(A):E|\rho) \nonumber \\
= &
\rE_\bX 
D( \sum_{i=1}^M  |i\rangle \langle i|  \otimes  \sum_{a:f_{\bX}(a)=i } P(a) \rho_{a}\| \frac{1}{M} I \otimes \rho^E) \nonumber \\
= &
\rE_\bX 
\sum_a \Tr P(a) \rho_{a} (\log (\sum_{a':f_{\bX}(a')=f_{\bX}(a)} P(a') \rho_{a'} ) -\log \frac{1}{M}\rho^E ) \nonumber \\
\le &
\sum_a P(a) \Tr \rho_{a} (\log (\rE_\bX \sum_{a':f_{\bX}(a')=f_{\bX}(a)} P(a') \rho_{a'} ) -\log \frac{1}{M}\rho^E ) 
\Label{8-15-a}
\\
= &
\sum_a P(a) \Tr \rho_{a} (\log (P(a) \rho_{a} \nonumber \\
& + \rE_\bX \sum_{a':f_{\bX}(a')=f_{\bX}(a),a'\neq a} P(a') \rho_{a'} ) -\log \frac{1}{M}\rho^E ) \nonumber \\
\le &
\sum_a P(a) \Tr \rho_{a} (\log (P(a) \rho_{a}\nonumber \\
& + \frac{1}{M}\sum_{a': a'\neq a} P(a') \rho_{a'} ) -\log \frac{1}{M}\rho^E ) 
\Label{8-15-b}
\\
\le &
\sum_a P(a) \Tr \rho_{a} (\log (P(a) \rho_{a}+ \frac{1}{M}\rho^E ) -\log \frac{1}{M}\rho^E ) \nonumber \\
\le &
\sum_a P(a) \Tr \rho_{a} (\log (v P(a) {\cal E}_{\rho^E}(\rho_{a})+ \frac{1}{M}\rho^E ) -\log \frac{1}{M}\rho^E ) 
\Label{8-15-c}
\\
= &
\sum_a P(a) \Tr \rho_{a} \log ( v M P(a) {\cal E}_{\rho^E}(\rho_{a})(\rho^E)^{-1}+ I) \nonumber ,
\end{align}
where
(\ref{8-15-a}) follows from the matrix convexity of $x \mapsto \log x$,
(\ref{8-15-b}) follows from Condition \ref{C1} and the matrix monotonicity of $x \mapsto \log x$,
and 
(\ref{8-15-c}) follows from 
(\ref{8-15-23}) and the matrix monotonicity of $x \mapsto \log x$.

Using Lemma \ref{L10}, we obtain
\begin{align}
& \sum_a P(a) \Tr \rho_{a} \log ( v M P(a) {\cal E}_{\rho^E}(\rho_{a})(\rho^E)^{-1}+ I) \nonumber \\
\le &
\frac{1}{s} \sum_a P(a) \Tr \rho_{a} ( v M P(a) {\cal E}_{\rho^E}(\rho_{a})(\rho^E)^{-1})^s \nonumber \\
= &
\frac{v^s M^s}{s} \sum_a P(a)^{1+s} \Tr {\cal E}_{\rho^E}(\rho_{a})^{1+s}(\rho^E)^{-s} \nonumber \\
= &
\frac{v^s M^s}{s} e^{s H_{1+s}(A|E| {\cal E}_{I\otimes \rho^E}(\rho ))}
\le 
\frac{v^s M^s}{s} e^{s H_{1+s}(A|E| \rho )},\Label{8-15-11}
\end{align}
where (\ref{8-15-11}) follows from (\ref{8-15-12}).

\section{Proof of Theorem \ref{th2}}\Label{a2}
The relations (\ref{8-15-14}) and (\ref{8-15-17})
imply
\begin{align*}
s \underline{I}'(A:E|\rho) 
\le 
\log \sum_a P(a) \Tr \rho_a (|{\cal A}| P(a) (\rho^E)^{-1/2} \rho_a (\rho^E)^{-1/2})^{s} .
\end{align*}
Therefore,
\begin{align}
& \rE_\bX e^{s \underline{I}'(f_{\bX}(A):E|\rho)} \nonumber \\
\le &
\rE_\bX 
\sum_a P(a) 
\Tr \rho_a (M  (\rho^E)^{-1/2} (\sum_{a':f_{\bX}(a')=f_{\bX}(a)} P(a') \rho_{a'} ) (\rho^E)^{-1/2})^{s} \nonumber \\
\le &
\sum_a P(a) \Tr \rho_a 
(M  (\rho^E)^{-1/2} \rE_\bX (\sum_{a':f_{\bX}(a')=f_{\bX}(a)} P(a') \rho_{a'} ) (\rho^E)^{-1/2})^{s} 
\Label{8-15-d} \\
= &
\sum_a P(a) \Tr \rho_a (M  (\rho^E)^{-1/2} (P(a) \rho_{a}
+ \rE_\bX (\sum_{a':f_{\bX}(a')=f_{\bX}(a),a \neq a'} P(a') \rho_{a'}) ) (\rho^E)^{-1/2})^{s} \nonumber \\
\le &
\sum_a P(a) \Tr \rho_a (M  (\rho^E)^{-1/2} (P(a) \rho_{a}
+ \frac{1}{M} (\sum_{a':a \neq a'} P(a') \rho_{a'}) ) (\rho^E)^{-1/2})^{s} 
\Label{8-15-e} \\
\le &
\sum_a P(a) \Tr \rho_a (M  (\rho^E)^{-1/2} (P(a) \rho_{a}+ \frac{1}{M}\rho^E ) (\rho^E)^{-1/2})^{s} 
\nonumber \\
= &
\sum_a P(a) \Tr \rho_a (I+M P(a) (\rho^E)^{-1/2}\rho_{a}(\rho^E)^{-1/2} )^s \nonumber \\
\le &
\sum_a P(a) \Tr \rho_a (I+M^s P(a)^s ((\rho^E)^{-1/2} \rho_{a}(\rho^E)^{-1/2} )^s) \Label{8-15-1}\\
= &
1+ M^s \sum_a P(a)^{1+s} \Tr \rho_a ((\rho^E)^{-1/2} \rho_{a}(\rho^E)^{-1/2} )^s) \nonumber \\
=&
1+ M^s e^{- s\overline{H}^*_{1+s}(A|E|\rho)}\nonumber 
\end{align}
where
(\ref{8-15-d}) follows from the matrix convexity of $x \mapsto x^s$,
and
(\ref{8-15-e}) follows from Condition \ref{C1} and the matrix monotonicity of $x \mapsto x^s$,
and 
(\ref{8-15-1}) follows from Lemma \ref{L9}.

\bibliographystyle{IEEE}

\end{document}